\documentclass[12pt]{svmult}
\usepackage{amsfonts} 
\usepackage{amsmath}
\usepackage{bm}

\usepackage{mathptmx}       
\usepackage{helvet}         
\usepackage{courier}        
\usepackage{type1cm}        
\usepackage{makeidx}         
\usepackage{graphicx}        
\usepackage{multicol}        
\usepackage[bottom]{footmisc}

\makeindex             

\def \bell {\mbox{\boldmath $\ell$}}

\begin{document}
\title*{Symplectic Dilations, Gaussian States and Gaussian Channels}
\author{K. R. Parthasarathy}
\institute{K. R. Parthasarathy \at Indian Statistical Institute, Delhi 
Centre, 7, S. J. S. Sansanwal Marg, New Delhi 110016, 
\email{krp@isid.ac.in}}
%
%
\maketitle
\vskip0.1in
\begin{center}
{\it Dedicated to Professor Kalyan B. Sinha on his 70th birthday}
\end{center}
\vskip0.1in

\abstract{By elementary matrix algebra we show that every real $2n 
\times 2n$ 
matrix admits a dilation to an element of the real symplectic group $Sp 
(2(n+m))$ for some nonnegative integer $m.$ Our methods do not yield the 
minimum value  of $m,$ for which such a dilation is possible.\\
After listing some of the main properties of Gaussian states in $L^2 
(\mathbb{R}^n),$  we analyse the implications of symplectic dilations in the 
study of quantum Gaussian channels which lead to some interesting open 
 problems, particularly, in the context of the work of Heinosaari, Holevo and 
Wolf \cite{3}.}


\section{Symplectic Dilation Theorem}
\label{sec:1}

Throughout this section we shall deal with matrices having real entries. Denote 
by $M_n$ the real linear space of all $n \times n$ matrices. Write
\begin{eqnarray}
  J_2 &= \left [ \begin{array}{cc} 0 & -1 \\ 1 & 0 \end{array} \right ], 
 \nonumber \\
  J_{2n} &= \left [ \begin{array}{ccccc} J_2 & 0 & 0 & \ldots & 0 \\ 0 & J_2 & 
0 & \ldots & 0 \\ \ldots & \ldots & \ldots & \ldots & \ldots \\
0 & 0 & \ldots & 0 & J_2  
\end{array}\right ] \label{eq1.1}
\end{eqnarray}
where the right hand side is a block diagonal matrix with diagonal blocks $J_2$ 
and the nondiagonal entries are $2 \times 2$ zero matrices. Let
\begin{equation}
Sp (2n) = \left \{ L \left | L \in M_{2n},  \right .  L^{T} J_{2n} L = 
J_{2n}   \right \}.\label{eq1.2}
\end{equation}
be the real { \it symplectic Lie group} of order $2n,$ where the index $T$ 
stands for transpose. It is known that every element $L$ in $Sp (2n)$ has 
determinant unity and has the property that $L^T \in Sp (2n).$ Any element of 
$Sp (2n)$ will be called a {\it symplectic} matrix of order $2n.$

The importance of $Sp(2n)$ in quantum probability lies in the property that the 
linear transformation
$$L \left [\begin{array}{c} p_1 \\ q_1 \\ \vdots \\ p_n \\ q_n  \end{array} 
\right ]  = \left [\begin{array}{c} p_1^{\prime} \\ q_1^{\prime} \\ \vdots \\ 
p_n^{\prime} \\ q_n^{\prime}  \end{array} 
\right ] $$
of the canonical position and momentum observables in $L^2 (\mathbb{R}^n)$ 
preserves the Heisenberg canonical commutation relations (CCR) if and only if 
$L \in Sp(2n).$

The main aim of this section is to establish a dilation theorem according to 
which any real linear transformation $A$ of the canonical momentum and position 
observables $\{p_r, q_r, r =1, 2, \ldots, n \}$ can be dilated to a symplectic 
linear transformation of a larger system $\{p_r, q_r, r=1,2, \ldots, n+m \}$ of 
such canonical observables obeying CCR. This is the essence of the following 
theorem in linear algebra whose proof will be accomplished by examining several 
special cases.

\begin{theorem} \label{thm1.1}
 Let $A \in M_{2n}.$ Then there exists a nonnegative integer $m,$ not depending 
on $A$ and matrices $B, C, D$ of respective order $2n \times 2m,$ $2m \times 
2n,$ $2m \times 2m$ such that the block matrix
\begin{equation}
\widetilde{A} = \left [ \begin{array}{cc}   A & B \\ C & D \end{array}  
\right ] \label{eq1.3}
\end{equation}
is a symplectic matrix of order $2 (n+m).$

\end{theorem}

\begin{lemma} \label{lem1.2}
 The following matrices of order $4 \times 4$ are symplectic.
\begin{itemize}
 \item[(i)] $\left [\begin{array}{cc} 0 & B \\ C & 0 \end{array}\right ]$ where 
$B,C \in SL (2, \mathbb{R}).$
\item[(ii)] $\left [\begin{array}{cccc} \alpha & 0 & 0 & x \\ 0&0& y &0 \\ 0 & 
x & \alpha & 0 \\ y & 0 & 0 & 0  \end{array}\right ]$ where $\alpha \neq 0,$ 
$1+xy = 0.$
\item[(iii)]$\left [\begin{array}{cccc} \alpha & 0 & x &  0\\ 0& \alpha& 0 &y 
\\ -y & 0 & \alpha & 0 \\ 0 & -x & 0 & \alpha  \end{array}\right ]$ where $xy = 
1 - \alpha^2.$
\item[(iv)]$\left [\begin{array}{cccc}\alpha & 0 & x &  0\\ 0& -\alpha& 0 &y 
\\ y & 0 & \alpha & 0 \\ 0 & x & 0 & -\alpha  \end{array}\right ]$ where $xy 
= 1+\alpha^2.$
\end{itemize}
\end{lemma}

\begin{proof}
 Straightforward verificaion. \hfill{$\qed$}
\end{proof}

\begin{lemma} \label{lem1.3}
 Theorem \ref{thm1.1} holds for $n=1$ with $m=1.$
\end{lemma}

\begin{proof}
We make a general observation that if an $2n \times 2n$ matrix $A$ satisfies 
Theorem  \ref{thm1.1}, then so does any matrix $L_1 A L_2$ where $L_1$ and 
$L_2$ are arbitrary elements of $Sp (2n).$ Since $Sp(2) = SL (2, \mathbb{R})$ 
and for any $2 \times 2$ matrix $A$ there exist $L_1, L_2$ in $SL (2, 
\mathbb{R})$ such that $L_1 AL_2$  has one of the following forms:
$$\left [ \begin{array}{cc} 0 & 0 \\ 0 & 0 \end{array} \right ],\left [ 
\begin{array}{cc} \alpha & 0 \\ 0 & 0 \end{array} \right ], \left [ 
\begin{array}{cc} \alpha & 0 \\ 0 & \alpha \end{array} \right ], \left [ 
\begin{array}{cc} \alpha & 0 \\ 0 & -\alpha \end{array} \right ] $$
where $\alpha \neq 0$ is a real scalar. By Lemma \ref{lem1.2} each of the four 
matrices above satisfies Theorem \ref{thm1.1} with $m=1.$ \hfill{$\qed$}
\end{proof}

\begin{definition}\label{def1.4}
Let   $A \in M_{2n}.$ If there exists a matrix $\widetilde{A} \in Sp (2 (n+m))$
such that (\ref{eq1.3})  holds we say that $\widetilde{A}$ is a  {\it 
symplectic dilation} of $A$ of {\it order} $2 (n+m).$
\end{definition}

\begin{lemma}\label{lem1.5}
If $A_i \in M_{2n_{i}}$  admits a symplectic dilation of order $2(n_i+m_i),$ 
$i=1,2,\ldots,k$ then their direct sum
$$\underset{i=1}{\bigoplus^{k}}\,\,\, A_i = \left [ \begin{array}{ccccc} A_1 & 
0 & 0 & \ldots & 0 \\ 0 & A_2 & 0 & \ldots & 0 \\ \ldots & \ldots & \ldots & 
\ldots & \ldots \\ 0 & 0 & 0 & \ldots & A_k \end{array} \right ]$$
admits a symplectic dilation of order $2 \sum\limits_{i=1}^{k} (n_i+m_i).$ 
\end{lemma}

\begin{proof}
 It is enough to prove for the case $k=2.$ Let
$$\widetilde{A}_i = \left [\begin{array}{cc} A_i & B_i \\ C_i & D_i \end{array} 
\right ] $$
be a symplectic dilation of $A_i$ of order $2 (n_i+m_i), i =1,2.$ Then
$$\widetilde{A}_1 \oplus \widetilde{A}_2 = \left [ \begin{array}{cccc}  A_1 & 
B_1 & 0 & 0 \\ C_1 & D_1 & 0 & 0 \\ 0 & 0 & A_2 & B_2 \\0 & 0 & C_2 & D_2    
\end{array} \right ] $$
is symplectic. First, interchange rows 2 and 3 followed by an interchange of 
columns 2 and 3 in order to obtain the symplectic matrix 
$$\left [ \begin{array}{cccc} A_1 & 0 & B_1 & 0 \\ 0 & A_2 &  0 & B_2 \\ C_1 & 0 
& D_1 & 0 \\ 0 & C_2 & 0 & D_2  \end{array} \right ]$$
of order $2(n_1 + m_1 + n_2 + m_2).$ This is clearly a symplectic dilation of 
$A_1 \oplus A_2.$ \hfill{$\qed$}
\end{proof}

\begin{lemma}\label{lem1.6}
Let $A_i,$ $1 \leq i \leq k$ be elements of $M_{2n},$ admitting symplectic 
dilations $\widetilde{A}_i$  of order $2(n+m)$  for each $i.$ If $p_i > 0,$ $1 
\leq i \leq k$ and $\sum\limits_{i=1}^{k} p_i = 1,$ then $\sum\limits_{i=1}^{k} 
p_i A_i$ admits a symplectic dilation of order $2k(n+m).$
\end{lemma}

\begin{proof}
 Consider the symplectic matrix $L = \underset{i=1}{\bigoplus^{k}} 
\,\widetilde{A}_i$ of order $2k(n+m).$ Let $((s_{ij}))$ be a real orthogonal 
matrix with its first row equal to $\left (\sqrt{p_1}, \sqrt{p_2}, 
\ldots,\sqrt{p_k}  \right ).$ Then
$$S = (( s_{ij} \,\,I_{2(n+m)}))$$
where $I_r$ denotes identity matrix of order $r,$ is a symplectic matrix of 
order $2(n+m)k.$ Thus $SLS^T$ is also symplectic. Considering this as a block 
matrix of the form (\ref{eq1.3}) we see that $SLS^T$ is a symplectic dilation 
of $\sum\limits_{i=1}^{k} p_i A_i.$ \hfill{$\qed$}
\end{proof}

\begin{lemma}\label{lem1.7}
Let $A$ be a real strictly positive definite matrix of order $2n.$ Then $A$ 
admits a symplectic dilation of order $4n.$
\end{lemma}

\begin{proof}
By Williamson's theorem \cite{1}, \cite{10} there exists a symplectic matrix 
$L$ 
of order $2n$ such that
$$L^T A L = \kappa_1 I_2 \oplus \kappa_2 I_2 \oplus \cdots \oplus \kappa_n I_2 $$
where $I_2$ is the identity matrix of order $2$ and $\kappa_1 \ge \kappa_2 \ge \cdots \ge 
\kappa_n > 0$ are the unique Williamson parameters of $A.$ By expression (iii) in 
Lemma \ref{lem1.2}, each $\kappa_j I_2$ has a symplectic dilation of order $4.$ Now 
Lemma \ref{lem1.5} implies that $L^TAL$ has a symplectic dilation of order 
$4n.$ \hfill{$\qed$} 
\end{proof}

\begin{lemma}\label{lem1.8}
Let $A$ be a symmetric matrix of order $2n.$ Then $A$ admits a symplectic 
dilation of order $8n.$ 
\end{lemma}

\begin{proof}
Choose and fix a constant $\lambda > 0$ so that $\lambda I + A$ and $\lambda I 
- A$ are both strictly positive definite. By Lemma \ref{lem1.7} both 
$\lambda I+A$ and $\lambda I-A$  admit symplectic dilations of order $4n.$ 
When $L$ is symplectic so is $-L$ and therefore $A - \lambda I$ has a 
symplectic dilation of order $4n.$ Now, by Lemma \ref{lem1.6}, $A = \frac{1}{2} 
(A+\lambda I + A - \lambda I)$ admits a symplectic dilation of order $8n.$ 
\hfill{$\qed$}
\end{proof}

\begin{lemma}\label{lem1.9}
Let $D,E,F,G$ be $2 \times 2$ matrices such that $\left [\begin{array}{cc} D & 
E \\ F & G \end{array} \right ]$ is symplectic of order $4.$ Then the $8 
\times 8$ matrix
$$\left [ \begin{array}{cccc} 0 & D & E & 0 \\ -D^T & 0 & 0 & -F^T \\ 0 & F & G 
& 0 \\ -E^T & 0 & 0 & -G^T    \end{array} \right ] $$
is symplectic. In particular, any $4 \times 4$ skew symetric matrix of the form 
$\left [\begin{array}{cc} 0 & D \\ -D^T & 0 \end{array}  \right ]$ admits a 
symplectic dilation of order $8.$ 
\end{lemma}

\begin{proof}
 Straightforward algebra. \hfill{$\qed$}
\end{proof}

\vskip0.1in
\noindent{\bf Proof of Theorem \ref{thm1.1}.} By Lemma \ref{lem1.3} it is 
enough to consider the case $n > 1.$   Express the $2n \times 2n$ 
matrix $A$ as an $n \times n$ block matrix
$$A = \left [ A_{ij} \right ], \quad i, j \in \left \{ 1,2,\ldots, n\right \}$$
where each $A_{ij}$ is of order $2 \times 2.$ Then $A$ can be written as
$$A = A_1 + A_2 + A_3$$
where $A_1 = \frac{1}{2} (A+A^T)$ is symmetric, $A_2$ is the block diagonal 
matrix 
$$\frac{1}{2} \left \{(A_{11} - A_{11}^T) \oplus (A_{22} - A_{22}^T) 
\oplus \cdots \oplus (A_{nn} - A_{nn}^T) \right \}$$ 
and 
$$A_3 = \underset{1 \leq i < j \leq n}{\sum} \,\,B_{ij}$$
with $B_{ij}$ being a block matrix with $ij$th block $\frac{1}{2} \left ( 
A_{ij} - A_{ji}^T \right )$ and $ji$th  block $\frac{1}{2} \left ( 
A_{ji} - A_{ij}^T \right )$ and all other blocks equal to $0.$  Write $A_1 = 
A_1^{+} - A_1^{-}$  where $A_1^{\pm},$ are positive and negative parts of 
$A_1.$ 
Put $B_1 = A_1^{+} + \varepsilon I,$  $B_2 = A_1^{-} + \varepsilon I$ where 
$\varepsilon > 0,$ so that $B_1$ and $B_2$ are strictly positive definite. Then 
$A_1 = B_1 - B_2$ where, by Lemma \ref{lem1.7}, $B_1$ and $B_2$ admit symplectic 
dilations of order $4n.$   To deal with the term $A_3$ we follow a method of 
colouring complete graphs \cite{8} as suggested by Ajit Iqbal Singh. First 
consider the case when $n$ is even. View the suffix pairs $ij$ of the matrices 
$B_{ij}$ as the edges of a complete graph $G=(V,E)$ where the vertex set $V$ is 
$\{1,2,\ldots, n\}$ and  the edge set $E$ is $\{ij, 1 \leq i < j \leq n \}.$ 
Suppose the vertices $2,3,\ldots, n$ are the (geometrical) vertices of a 
regular polygon of $n-1$ sides and the vertex $1$ is at the centre of the 
polygon. Denote by $E_r$ the set consisting of the edge $1 \,\,r+1$ and the 
$\frac{1}{2}n-1$ edges perpendicular to the edge  $1 \,\,r+1.$ Then no two 
edges in $E_r$ have a common vertex and no two $E_r$'s have common edge. Thus 
the set $E$ can be expressed as a disjoint union:
$$E = \underset{r=1}{\bigcup^{n-1}} E_r, \,\,E_r \cap E_s = \emptyset 
\,\,\mbox{if}\,\, r \neq s.$$
If 
$$C_r = \sum\limits_{ij \in E_{r}} \,B_{ij}$$
then
$$A_3 = \sum\limits_{r=1}^{n-1} \,C_r.$$
It is clear that $C_r$ is  a direct sum of $n/2$ matrices of order $4 \times 4:$
$$\left [ \begin{array}{cc} 0 & \frac{1}{2} \left ( A_{ij} - A_{ji}^T \right ) 
\\ \frac{1}{2} \left ( A_{ji} - A_{ij}^T \right ) & 0   \end{array} \right ], 
\,\, ij \in E_r.$$
By Lemma \ref{lem1.5} and  Lemma \ref{lem1.9}, each $C_r$ admits a symplectic 
dilation of order $8 \times n/2 = 4n.$ By the same arguments each of the $n+2$ 
matrices $(n+2)B_1,$ $-(n+2)B_2,$ $(n+2)A_2,$ $(n+2)C_1, \ldots, (n+2) C_{n-1}$ 
admits a symplectic dilation of order $4n.$ By Lemma \ref{lem1.6} the matrix
$$A = \frac{1}{n+2} \left \{ (n+2)B_1-(n+2)B_2+(n+2)A_2+(n+2)C_1 + \cdots 
+ (n+2) C_{n-1} \right \} $$
admits a symplectic dilation of order $4n (n+2).$ This completes the proof when 
$n$ is even.

Let now $n$ be odd. Then add one more vertex $n+1$ and form the edge sets $E_1, 
E_2, \ldots, E_n$ as before. Then each $E_r$  has $\frac{n+1}{2}$ edges and 
exactly one of them meets the vertex $n+1.$ Delete it and call the remaining 
subset $E_r^{\prime}$ which has $\frac{n-1}{2}$ edges. Now define
$$C_r = \sum_{ij \in E_{r}^{\prime}} \, B_{ij}, \quad 1 \leq r \leq n.$$
Now repeating the same arguments we get a symplectic dilation of $A$ of order 
$4n(n+3).$ \hfill{$\square$}

\begin{remark}
Theorem \ref{thm1.1} and its proof show that every real matrix of order $2n 
\times 2n$ admits a symplectic dilation of order $4n(n+2)$ if $n$ is even and 
$4n(n+3)$ if $n$ is odd. The problem of finding the size of the minimal 
symplectic dilation of a matrix of order $2n \times 2n$ remains open. It is to 
be noted that symplecticity of order $2k$ is with respect to the matrix 
$J_{2k}$ as defined in \eqref{eq1.1}.
\end{remark}

\section{The Symplectic Group, Weyl operators and Gaussian States}
\label{sec:2}

We shall present in this section a brief account of the role of the group $Sp 
(2n)$ in describing Gaussian states in $L^2 (\mathbb{R})$ and some of their 
properties which are important in the formulation of the notion of Gaussian 
channels in quantum information theory. For proofs of results stated here we 
refer to Holevo \cite{4} and Parthasarathy \cite{6}, \cite{7}.

Consider  the complex Hilbert space $\mathbb{C}^n$ and $L^2 (\mathbb{R}^n),$ 
the space of complex-valued, absolutely square integrable functions with 
respect to the $n$-dimensional Lebesgue measure. We shall write all scalar 
products in the Dirac notation. For any $ \mathbf{u} = (u_1, u_2, \ldots, 
u_n)^T$ in $\mathbb{C}^n,$ associate the {\it exponential vector} $e 
(\mathbf{u})$ in $L^2 (\mathbb{R}^n)$ by
\begin{equation}
 e (\mathbf{u})(\mathbf{x}) = (2 \pi)^{-n/4} \,\, \exp \,\sum_{j=1}^n \left ( 
u_j x_j - \frac{1}{2} u_j^2 - \frac{1}{4} x_j^2  \right )  \label{eq2.1}
\end{equation}
for $\mathbf{x} \in \mathbb{R}^n.$ Exponential vectors span a dense linear 
manifold $\mathcal{E} \subset L^2 (\mathbb{R}^n)$ and any finite number of 
exponential vectors are linearly independent. Furthermore,
\begin{eqnarray*}
 \langle e (\mathbf{u}) \big | e (\mathbf{v}) \rangle &=& \exp \langle 
\mathbf{u} | \mathbf{v} \rangle \\
&=& \exp \sum_{j=1}^{n} \overline{u}_j v_j.
\end{eqnarray*}
In particular,
$$| \psi (\mathbf{u}) \rangle = e^{-\frac{1}{2} \|\mathbf{u}\|^{2}} | e 
(\mathbf{u}) \rangle$$
is a unit vector and the pure state determined by this vector is called the 
{\it coherent state} associated with $\mathbf{u}.$

For any $\mathbf{u} \in \mathbb{C}^n,$ there exists (an exponential like) 
unitary operator $W(\mathbf{u})$ in $L^2 (\mathbb{R}^n),$ called {\it  Weyl operator} 
associated with $\mathbf{u},$ satisfying the relations
\begin{equation}
W(\mathbf{u}) | e (\mathbf{v}) \rangle = e^{- \frac{1}{2} \|\mathbf{u}\|^{2} - 
\langle\mathbf{u}|\mathbf{v} \rangle  }   \,\,\, | e (\mathbf{u}+\mathbf{v}) 
\rangle
\label{eq2.2}
\end{equation}
for all $\mathbf{v} \in \mathbb{C}^n.$ For a given  $\mathbf{u}$ such a 
$W(\mathbf{u})$ is uniquely defined. They obey the Weyl commutation relations
\begin{equation}
W(\mathbf{u}) W(\mathbf{v}) = e^{-i \mathcal{I}m \langle \mathbf{u}| \mathbf{v} 
\rangle} W (\mathbf{u}+\mathbf{v}).   \label{eq2.3}
\end{equation}
It is a multiplicative family of operators modulo a scalar multiplier. From 
(\ref{eq2.3}) one gets
\begin{equation}
W (\mathbf{u}) W(\mathbf{v}) W(\mathbf{u})^{-1} = e^{-2i \mathcal{I}m \langle 
\mathbf{u} | \mathbf{v} \rangle} W (\mathbf{v}).   \label{eq2.4}
\end{equation}
The correspondence $\mathbf{u} \rightarrow W(\mathbf{u})$ is a strongly 
continuous, projective, unitary and irreducible representation of the additive 
group $\mathbb{C}^n.$ If we identify $L^2 (\mathbb{R}^{n+m})$ with $L^2 
(\mathbb{R}^n) \otimes L^2 (\mathbb{R}^m)$ and $\mathbb{C}^{n+m}$ with 
$\mathbb{C}^n \oplus \mathbb{C}^m,$ then $W (\mathbf{u} \oplus \mathbf{v})$ 
gets identified with $W(\mathbf{u}) \otimes W(\mathbf{v})$ and we simply write 
$W (\mathbf{u} \oplus \mathbf{v}) = W (\mathbf{u}) \otimes W(\mathbf{v}).$ This 
is called the {\it factorizability} property of the Weyl representation $ 
\mathbf{u} \rightarrow W(\mathbf{u}).$

Suppose $\mathbf{u} \rightarrow W^{\prime} (\mathbf{u})$ is a strongly 
continuous map from $\mathbb{C}^n$ into the unitary group of a complex 
separable Hilbert space $\mathcal{H}$ satisfying equation (\ref{eq2.3}) with 
$W$ replaced by $W^{\prime}.$ Then, according to Stone-von Neumann theorem 
there is a unitary isomorphism $\Gamma$ from $\mathcal{H}$ onto $L^2 
(\mathbb{R}^n) \otimes k$ for some Hilbert space $k$ such that
$$\Gamma W^{\prime} (\mathbf{u}) \Gamma^{-1} = W (\mathbf{u}) \otimes I_k  
\quad \forall \quad \mathbf{u} \in \mathbb{C}^n,$$
$I_k$ being the identity opertor in $k.$ If $W^{\prime}$ is also  irreducible 
then $W$ and $W^{\prime}$ are unitarily equivalent.

For any $L \in Sp(2n)$ define the action of $L$ on $\mathbf{u}$ by $L \cdot 
\mathbf{u} = \mathbf{u}^{\prime}$ where
$$L \left [ \begin{array}{c} Re \, u_1 \\ \mathcal{I}m \, u_1 \\ \vdots \\ Re 
\,u_n \\ \mathcal{I}m \,u_n  \end{array} \right ] = \left [ \begin{array}{c}  Re 
\, u_1^{\prime}\\ \mathcal{I}m \, u_1^{\prime} \\ \vdots \\ Re 
\,u_n ^{\prime}\\ \mathcal{I}m \,u_n^{\prime} 
\end{array} \right ] \quad \forall \quad \mathbf{u} \in \mathbb{C}^n.$$
Then such an action preserves the real bilinear form $\mathcal{I}m \langle 
\mathbf{u}| \mathbf{v}\rangle$ and therefore
$$W (L \cdot \mathbf{u}) \,W(L \cdot \mathbf{v}) = e^{-i \,\mathcal{I}m \, 
\langle \mathbf{u}| \mathbf{v}\rangle} W (L \cdot (\mathbf{u}+\mathbf{v})). $$
Hence by the discussion above there exists  a unitary operator, say $\Gamma (L),$ 
such that
\begin{equation}
\Gamma(L) \, W(\mathbf{u}) \Gamma (L)^{-1} = W(L \cdot \mathbf{u}) \quad 
\forall \,\,\mathbf{u} \in \mathbb{C}^n.      \label{eq2.5}
\end{equation}
Since the projective representation $\mathbf{u} \rightarrow W (L \cdot 
\mathbf{u})$ is also irreducible it follows that the choice of the unitary 
operator $\Gamma(L)$ is unique upto a scalar multiple of modulus unity. Thus 
the map $L \rightarrow \Gamma(L)$ is a projective unitary representation of 
$Sp(2n).$ If $U$ is a unitary matrix of order $n \times n,$ $U$ preserves the 
scalar product in $\mathbb{C}^n$ and hence preserves $\mathcal{I}m \langle 
\mathbf{u}|\mathbf{v}\rangle.$ In other words there exists a matrix $L^U \in 
Sp(2n) \cap SO(2n)$ such that $L^U \cdot \mathbf{u} = U \mathbf{u}.$ There 
exists a unique unitary operator $\Gamma_0 (U)$ in $L^2 (\mathbb{R}^n)$ such 
that
$$ \Gamma_0 (U) e(\mathbf{u}) = e (U \mathbf{u}) \quad \forall \quad 
\mathbf{u} \in \mathbb{C}^n.$$
Then $\Gamma_0 (U) \Gamma_0(U^{\prime})=\Gamma_0 (UU^{\prime})$ for any 
$U,U^{\prime} \in U (n)$ and $\Gamma (L^U)$ can be chosen to be $\Gamma_0 (U).$ 
The representation $U \rightarrow \Gamma_0(U)$ of $U(n)$ is called the {\it 
second quantization} map. Thus we get a projective unitary representation $L 
\rightarrow \Gamma (L)$ of $Sp(2n)$ such that (\ref{eq2.5}) holds, $\Gamma 
(L_1) \Gamma(L_2) = \Gamma (L_1 L_2)$ whenever $L_1$ and $L_2$ are symplectic 
and orthogonal but, in general,
$$\Gamma (L_1) \Gamma(L_2) = \sigma (L_1, L_2) \Gamma (L_1 L_2)$$
where $\sigma(L_1, L_2)$ is a scalar of modulus unity.

\begin{definition}\label{def2.1}
A {\it state} $\rho$ in $L^2(\mathbb{R}^n)$ is a positive operator of unit 
trace and its {\it quantum Fourier transform} $\widehat{\rho}(\mathbf{u})$ is 
the complex-valued function on $\mathbb{C}^n$ defined by
$$\widehat{\rho} (\mathbf{u}) = \mbox{Tr}  \,\,\, \rho \,W(\mathbf{u}) \quad 
\forall \quad \mathbf{u} \in \mathbb{C}^n.$$ 
The quantum Fourier transform of a state satisfies the following properties:
\begin{itemize}
 \item[(i)] \quad $\widehat{\rho} (\mathbf{0}) = 1$ and the map $\mathbf{u} 
\rightarrow \widehat{\rho} (\mathbf{u})$ is continuous on $\mathbb{C}^n.$
 \item[(ii)]\quad The kernel
$$k_{\rho} (\mathbf{u}, \mathbf{v}) = e^{i \,\mathcal{I}m \langle 
\mathbf{u}|\mathbf{v}\rangle} \widehat{\rho} (\mathbf{v}-\mathbf{u}), 
\mathbf{u}, \mathbf{v} \in \mathbb{C}^n$$
is positive definite, i.e., for any finite set $\{\mathbf{u}_r, 1 \leq r \leq 
m\} \subset \mathbb{C}^n$ and scalars $c_r, 1 \leq r \leq m$
$$\Sigma \overline{c}_r c_s \,\,k_{\rho} (\mathbf{u}_r, \mathbf{u}_s) \ge 0.$$
 \item[(iii)]\quad (Inversion formula) For any state $\rho$ in $L^2 
(\mathbb{R}^n)$
\begin{eqnarray*}
 \rho &=& \frac{1}{\pi^{n}} \,\,\int\,\,\widehat{\rho}(\mathbf{u})
\,\,W(-\mathbf{u}) d^{2n} \,\mathbf{u}\\
&=& \frac{1}{\pi^{n}} \,\,\int\,\,\overline{\widehat{\rho}(\mathbf{u})} 
\,\,W(\mathbf{u}) d^{2n} \,\mathbf{u},
\end{eqnarray*}
where $d^{2n} \mathbf{u}$ is the $2n$-dimensional Lebesgue measure when 
$\mathbb{C}^n$ is considered as the real linear space $\mathbb{R}^{2n}.$
 \item[(iv)] \quad (Quantum Bochner's Theorem) let $\varphi$ be any 
complex-valued continuous function on $\mathbb{C}^n$ such that $\varphi 
(\mathbf{0})=1$ and the kernel $k(\mathbf{u}, \mathbf{v}) = \varphi 
(\mathbf{v}-\mathbf{u}) \exp \,\,i \langle \mathbf{u}| \mathbf{v}\rangle$ is 
positive definite. Then there exists a unique state $\rho$ in $L^2 
(\mathbb{R}^n)$ such that $\widehat{\rho}=\varphi$
 \item[(v)]\quad For any state $\rho$ in $L^2 (\mathbb{R}^n)$ and any $L \in Sp(2n)$
$$\left [ \Gamma(L) \,\,\rho\,\,\Gamma(L)^{-1} \right ]^{\wedge} (\mathbf{u}) 
= \widehat{\rho} (L^{-1} \mathbf{u}) \quad \forall \quad \mathbf{u} \in 
\mathbb{C}^n.$$
 \item[(vi)] \quad For any state $\rho$ in $L^2 (\mathbb{R}^{n+m}) = L^2 
(\mathbb{R}^n) \otimes L^2 (\mathbb{R}^m)$ define the {\it marginal} states 
$\rho_1$ and $\rho_2$ in $L^2 (\mathbb{R}^n)$ and $L^2 (\mathbb{R}^m)$ 
respectively by
$$\rho_1 = \mbox{Tr}_2\,\, \rho, \quad \rho_2 = \mbox{Tr}_1\,\,\rho  $$
where $\mbox{Tr}_1$ and $\mbox{Tr}_2$ are the relative traces of $\rho$ over 
the factors $L^2 (\mathbb{R}^n)$ and $L^2 (\mathbb{R}^m)$ respectively. Then
\begin{eqnarray*}
 \widehat{\rho}_1 (\mathbf{u}) &=&  \widehat{\rho} (\mathbf{u} \oplus 
\mathbf{0}),\\
\widehat{\rho}_2 (\mathbf{v}) &=&  \widehat{\rho} (\mathbf{0} \oplus 
\mathbf{v})
\end{eqnarray*}
where $\mathbf{u} \in \mathbb{C}^n,$ $\mathbf{v} \in \mathbb{C}^m.$
\end{itemize}
\end{definition}

\begin{definition}\label{def2.2}
 A state $\rho$ in $L^2 (\mathbb{R}^n)$ is said to be {\it Gaussian} if its  
quantum Fourier transform has the form
$$\widehat{\rho} (\mathbf{u}) = \exp \,\,P(x_1, y_1, \ldots, x_n, y_n) \quad 
\forall \quad \mathbf{u} \in \mathbb{C}^n$$
where $x_j = \mbox{Re}\, u_j,$ $y_j = Im \,\,u_j$ and $P$ is a 
polynomial in $2n$ variables of degree at most 2.
\end{definition}

\begin{theorem}\label{thm2.3}
A state $\rho$ in $L^2 (\mathbb{R}^n)$ is Gaussian if and only there exist 
vectors ${\boldsymbol \ell},$ $\mathbf{m}$ in $\mathbb{R}^n$ and 
a 
real positive 
definite matrix $S$ of order $2n$ satisfying the inequality
\begin{equation}
2 S - i \,\,J_{2n} \ge 0,    \label{eq2.6}
\end{equation}
such that the quantum Fourier transform $\widehat{\rho} (\mathbf{u})$ is given 
by
\begin{equation}
\widehat{\rho} (\mathbf{u}) = \exp - i \sqrt{2}  \left ( {\boldsymbol 
\ell}^T \,\,\mathbf{x} - 
\mathbf{m}^T \mathbf{y} \right ) - \left (\mathbf{x}^T, \mathbf{y}^T \right ) S 
\binom{\mathbf{x}}{\mathbf{y}}  \quad \forall \,\,\mathbf{u} \in 
\mathbb{C}^{2n} \label{eq2.7}
\end{equation}
where $x_j = \mbox{Re}\, u_j,$ $y_j = Im \,u_j,$ $1 \leq j \leq n$ 
and $J_{2n}$ is given by (\ref{eq1.1}). 
\end{theorem}

\begin{proof}
 See \cite{6}.
\end{proof}

\begin{remark}
We shall give now an interpretation of the parameters $\bell,$ $\mathbf{m}$ and 
$S$ occuring in  (\ref{eq2.7}) and also the matrix inequality (\ref{eq2.6}) in 
the language of the momentum and position observables obeying CCR. To this end 
we first observe that (\ref{eq2.3}) implies that for any fixed $\mathbf{u},$ 
$\{W (t \mathbf{u}), t \in \mathbb{R}\}$ is a strongly continuous one parameter 
 group of unitary operators admitting a self adjoint Stone generator 
$p(\mathbf{u})$ so that
$$W (t \mathbf{u}) = e^{-it p(\mathbf{u})}, \mathbf{u} \in \mathbb{C}^n, t \in 
\mathbb{R}. $$
Writing $\mathbf{e}_j = (0, \ldots, 0, 1, 0, \ldots, 0)^T$ for $1 \leq j \leq 
n,$ where 1 occurs in the $j$th position and putting
$$p_j = 2^{-\frac{1}{2}} \,\, p(e_j), \,\, q_j = -2^{-\frac{1}{2}} p (i e_j) $$
we obtain selfadjoint operators obeying CCR on the linear manifold 
$\mathcal{E}$ generated by the exponential vectors. Then $W(\mathbf{u})$ can be 
expressed as
$$W (\mathbf{u}) = \exp - i \sqrt{2} \sum_{j=1}^n \left ( x_j p_j - y_j 
q_j\right )$$ 
with $x_j = \mbox{Re}\, u_j,$ $y_j = \mathcal{I}m \, u_j$ for every $j$ and 
$\sum_{j=1}^n \left ( x_j p_j - y_j q_j\right )$ being the selfadjoint closure 
from $\mathcal{E}.$ From the definition of quantum Fourier transform and 
(\ref{eq2.7}) it follows that each observable $p (\mathbf{u})$ has a normal 
distribution in the Gaussian state $\rho$ with characteristic function
$$\exp - i t \sqrt{2} \sum_{j=1}^n \left ({\boldsymbol \ell}_j x_j - m_j 
y_j \right ) - t^2 
\mathbf{\xi}^T S \mathbf{\xi}, \quad t \in \mathbb{R} $$
where $\mathbf{\xi}^T = (x_1, y_1, x_2, y_2, \ldots, x_n, y_n).$ Thus
$${\boldsymbol \ell}_j = \mbox{Tr}\, \rho p_j, \quad m_j = \mbox{Tr}\, \rho 
\, q_j$$
and the covariance matrix of $(p_1, -q_1, \ldots, p_n, -q_n)$ is $S.$ When 
(\ref{eq2.7}) holds we write
\begin{equation}
\rho = \rho_g ({\boldsymbol \ell}, \mathbf{m}, S)\label{eq2.8} 
\end{equation}
in order to indicate that $\rho$ is a Gaussian state with {\it mean} momentum 
and position vectors ${\boldsymbol \ell}, \mathbf{m}$ respectively 
and {\it 
covariance matrix} $S.$ If we write
$$(Z_1, Z_2, \ldots, Z_{2n}) = (p_1, -q_1, \ldots, p_n, - q_n)$$
then
$$\mbox{Cov} (Z_r, Z_s) = \mbox{Tr}\, \frac{Z_r Z_s+Z_s Z_r }{2} \rho - 
(\mbox{Tr} \,Z_r \rho)(\mbox{Tr}\, Z_s \rho) $$
is the $rs$-th entry of $S.$ With this convention the inequality (\ref{eq2.6}) 
encapsulates the uncertainty principle for all momentum - position pairs 
$(P,Q)$ where
\begin{eqnarray*}
Q &=& \sum_{r=1}^n \,\,\left (x_r p_r - y_r q_r \right ) \\
P &=& \sum_{r=1}^n \,\,\left (x^{\prime}_r p_r - y^{\prime}_r q_r \right ) 
\end{eqnarray*}
and $\sum\limits_{r=1}^n \left ( x_r y^{\prime}_r - x^{\prime}_r y_r \right ) = 
1,$ $x_r,$ $y_r,$ $x_r^{\prime}, y_r^{\prime}$ being real scalars.
\end{remark}

We now enumerate some of the basic properties of Gaussian states, parametrized 
as in (\ref{eq2.8}). They are essentially corollaries of Theorem \ref{thm2.3} 
and the basic properties of quantum Fourier transform mentioned earlier.
\begin{itemize}
 \item[(i)]\quad Tensor products and marginals of Gaussian states are Gaussian.

 \item[(ii)] \quad (Tranformation property) For any $\mathbf{u} = \mathbf{s}+i 
\mathbf{t}, \mathbf{s}, \mathbf{t} \in \mathbb{R}^n,$
$$W (\mathbf{u}) \rho_g ({\boldsymbol \ell}, \mathbf{m}, S) 
W(\mathbf{u})^{-1} = 
\rho_g ({\boldsymbol \ell} + \sqrt{2} \mathbf{t}, \mathbf{m} + 
\sqrt{2}
\mathbf{s}, S), $$
and for any $L \in Sp (2n),$
$$\Gamma (L) \rho_g ({\boldsymbol \ell}, \mathbf{m}, S) 
\Gamma(L)^{-1} = \rho_g 
({\boldsymbol \ell}^{\prime}, \mathbf{m}^{\prime}, S^{\prime})$$
where
\begin{eqnarray*}
 \left [ \begin{array}{c} \ell^{\prime}_1 \\ -m_1^{\prime} \\ \vdots \\ 
\ell_n^{\prime} \\ -m_n^{\prime}  \end{array} \right ] &=& (L^{-1})^T \left [ 
\begin{array}{c}  \ell_1 \\ -m_1\\ \vdots \\ 
\ell_n \\ -m_n  \end{array} \right ], \\
S^{\prime} &=& (L^{-1} )^T S L^{-1}
\end{eqnarray*}

\item[(iii)] \quad If $\Gamma$ is any unitary operator in $L^2 (\mathbb{R}^n)$ 
satisfying the property that $\Gamma \rho \Gamma^{-1}$ is a Gaussian state 
whenever $\rho$ is a Gaussian state in $L^2 (\mathbb{R}^n)$ then $\Gamma$ is 
given by
\begin{equation}
\Gamma = \lambda W (\mathbf{u}) \Gamma (L)    \label{eq2.9}
\end{equation}
for some complex scalar $\lambda$ of mudulus unity, $\mathbf{u} \in 
\mathbb{C}^n$ and $L \in Sp(2n).$
 
In view of properties (ii) and (iii) any unitary operator $\Gamma$ of the form 
(\ref{eq2.9}) is called a {\it Gaussian symmetry operator}. All Gaussian 
symmetry operators in $L^2 (\mathbb{R}^n)$ constitute a group $\mathcal{G}_n.$

Suppose $U$ is a unitary operator in $L^2 (\mathbb{R}^n)$ such that for every 
pure Gaussian state $| \psi \rangle$ the pure state $U \,\,| \psi \rangle$ is 
also Gaussian. Is $U$ a Gaussian symmetry operator? We do not know the answer.

\item[(iv)]\quad  The covariance matrix $S$ in $\rho_g ({\boldsymbol 
\ell}, 
\mathbf{m}, S)$ admits a  representation
$$ S = L^T (\kappa_1 I_2 \oplus \kappa_2 I_2 \oplus \cdots \oplus \kappa_n I_2) L$$
for some $L \in Sp (2n),$ $\kappa_1 \ge \kappa_2 \ge \cdots \ge \kappa_n \ge \frac{1}{2}.$ In 
such a representation $L$ need not be unique but $\kappa_1, \kappa_2, \ldots, \kappa_n$ are 
unique. The $\kappa_j$'s are the {\it Williamson parameters} of $S.$

\item[(v)]\quad The covariance matrix $S$ of a Gaussian state in $L^2 
(\mathbb{R}^n)$ admits a representation
$$S = \frac{1}{4} \left ( L^T L + M^T M \right )$$
for some $L, M$ in $Sp(2n).$ Such a representation need not be unique. In the 
convex set $\mathcal{K}_n$ of all Gaussian covariance matrices of order $2n,$ 
an element 
$S$ is an extreme point if and only if $S=\frac{1}{2} L^T L$ for some $L$ in 
$Sp(2n).$ A Gaussian state $\rho$ in $L^2 (\mathbb{R}^n)$ is pure if and only 
if its covariance matrix is of the form $\frac{1}{2} L^TL$ for some $L$ in $Sp 
(2n)$ and in such a case its wave function $|\psi\rangle$ has the form
$$| \psi \rangle = W (\mathbf{u}) \Gamma(L) \,\,|e (\mathbf{0}) \rangle. $$
Every coherent state is Gaussian with covariance matrix $\frac{1}{2} I_{2n}.$

\item[(vi)] \quad(Gaussian purification property) Let $\rho_g ({\boldsymbol 
\ell}, 
\mathbf{m}, S)$ be a Gaussian state in $L^2 (\mathbb{R}^n)$ with
$$S = \frac{1}{4} \left (L_1 ^T L_1 + L_2^T L_2  \right ) $$
as in property (v). Put
$$| \psi_j \rangle = \Gamma (L_j)^{-1} |e(\mathbf{0}) \rangle, \quad j=1,2, $$
and define the second quantization unitary operator $\Gamma_0$ in $L^2 
(\mathbb{R}^n) \otimes L^2 (\mathbb{R}^n)$ by the relations
$$\Gamma_0 e (\mathbf{v} \oplus \mathbf{v}^{\prime}) = e \left 
(\frac{\mathbf{v}+\mathbf{v}^{\prime}}{\sqrt{2}} \oplus  
\frac{\mathbf{v}-\mathbf{v}^{\prime}}{\sqrt{2}}\right ) \quad \forall \quad 
\mathbf{v}, \mathbf{v}^{\prime} \in \mathbb{C}^n.$$
Putting
$$ \Gamma = \left (W \left ( \frac{\mathbf{m} + i {\boldsymbol \ell}}{\sqrt{2}} 
\right ) \otimes I \right ) \Gamma_0$$
one gets a purification of $\rho_g ({\boldsymbol \ell}, \mathbf{m}, S)$ as
$$\rho_g ({\boldsymbol \ell}, \mathbf{m}, S) = \mbox{Tr}_2\, \Gamma \left 
(|\psi_1 
\rangle \langle \psi_1 | \otimes |\psi_2 \rangle \langle \psi_2| \right ) 
\Gamma^{-1}$$
where the pure state within $\mbox{Tr}_2$ is also Gaussian.

\item[(vii)]\quad (von Neumann entropy of $\rho_g ({\boldsymbol \ell}, 
\mathbf{m}, 
S)$)\\First write
$$\rho_g ({\boldsymbol \ell}, \mathbf{m},S) = W \left (\frac{\mathbf{m}+i 
{\boldsymbol \ell}}{\sqrt{2}} \right ) \rho_{g} (0,0,S) W \left 
(\frac{\mathbf{m}+i 
{\boldsymbol \ell}}{\sqrt{2}} \right )^{-1} $$
using transformaton property (ii). Using property (iv) and the Williamson 
parameters $\kappa_1 \ge \kappa_2 \ge \cdots \ge \kappa_n \ge \frac{1}{2}$ we express
$$\rho_g ({\boldsymbol \ell}, \mathbf{m}, S) = W \left ( \frac{\mathbf{m} + i 
{\boldsymbol \ell}}{\sqrt{2}} \right ) \Gamma (L)^{-1} 
\underset{j=1}{\bigotimes^{r}} \,\rho_g (\mathbf{0}, \mathbf{0}, \kappa_j I_2)
\bigotimes \rho_g (0, 0, \frac{1}{2} I_2)^{\otimes n-r} \,\,
\Gamma (L) W \left (\frac{\mathbf{m} + i 
{\boldsymbol \ell}}{\sqrt{2}} \right )$$
where we assume 
\begin{eqnarray*}
 \kappa_j &>& \frac{1}{2} \quad {\rm if}  \quad  1 \leq j \leq r,\\
&=& \frac{1}{2} \quad {\rm if}  \quad r + 1 \leq j \leq n.
\end{eqnarray*}
%
\end{itemize}
In $L^2 (\mathbb{R})$ denote the momentum and position operators by $p$ and $q$ 
respectively and note that
\begin{equation}
\rho_g (\mathbf{0}, \mathbf{0}, \kappa I_2) = \left \{\begin{array}{l} |e 
(\mathbf{0}) \rangle \langle e (\mathbf{0}| \quad {\rm if} \quad \kappa=\frac{1}{2}, 
\\
(1-e^{-s} e^{-\frac{1}{2}s(p^{2} + q^{2}-1)}) \quad {\rm if}  \quad \kappa > 
\frac{1}{2}  \end{array} \right .   \label{eq2.10}
\end{equation}
where $s$ is given by $\kappa=\frac{1}{2} \coth \frac{1}{2}s$ with $s>0.$ Writing
$\kappa_j = \frac{1}{2} \coth \frac{1}{2} s_j, \quad s_j > 0, \quad 1 \leq j \leq r$
we see that $\rho_g ({\boldsymbol \ell}, \mathbf{m}, S)$ is unitarily equivalent 
to 
the tensor product state
$$\underset{j=1}{\bigotimes^{r}}  \left (1-e^{-s_{j}} \right ) e^{-\frac{1}{2} 
s_{j} (p_{j}^{2} + q_{j}^{2}-1)} \otimes(|e (0)\rangle \langle e(0) |)^{\otimes 
n-r}. $$
In particular, every Gaussian state is conjugate by a Gaussian symmetry 
operator to a product of $r$ thermal states and $(n-r)$ vaccuum states in 
$L^2(\mathbb{R})$ where $0 \leq r \leq n.$

The von Neumann {\it entropy} of  a state $\rho$ denoted by $S(\rho)$ is 
defined as the quantity $-\mbox{Tr} \,\,\rho \,\,\log\,\,\rho.$ If $\rho$ has 
eigenvalues $\lambda_1, \lambda_2, \ldots$ inclusive of multiplicity then $S 
(\rho) = - \sum\limits_j \lambda_j \,\log \,\lambda_j.$ Thus $S(U \rho 
U^{-1})=S(\rho)$ for any unitary operator $U$ and $S(\rho_1 \otimes \rho_2) = S 
(\rho_1) + S(\rho_2)$ for any product state $\rho_1 \otimes \rho_2.$ Thus 
$$S (\rho_g ({\boldsymbol \ell}, \mathbf{m}, S)) = \sum_{j=1}^r \,\,\, S 
(\rho_g 
(0, 0, \kappa_j I_2)). $$
Since the number operator $\frac{1}{2} (p^2+q^2-1)$ has eigenvalues 
$0,1,2,\ldots$ with multiplicity $1$ each, it follows from (\ref{eq2.10}) that
$$S (\rho_g (0,0,\kappa I_2)) = \left \{\begin{array}{l} 0 \quad {\rm if} \quad \kappa = 
\frac{1}{2}, \\ \frac{2}{2\kappa+1} H \left (\frac{2\kappa-1}{2\kappa+1} \right ) \quad {\rm 
if} \quad   \kappa > \frac{1}{2}\end{array} \right . $$
where $H$ is the Shannon entropy function given by $H(t) = -t \,\log\,t-(1-t) 
\,\log \, (1-t), \,0 \leq t \leq 1$ with $H(0) =  H(1)=0.$ Thus
$$S (\rho_g ({\boldsymbol \ell}, \mathbf{m}, S)) = \sum_{j=1}^n  
\,\frac{2}{2\kappa_j+1} 
H \left (\frac{2\kappa_j-1}{2\kappa_j+1} \right ) $$
where $\kappa_1 \ge \kappa_2 \ge \cdots \ge \kappa_n \ge \frac{1}{2}$ are the Williamson 
parameters of $S.$

\section{Gaussian Channels}
 A {\it quantum channel} is a completely positive, trace preserving and linear 
map on the algebra $\mathcal{B}(\mathcal{H})$ of all bounded operators on a 
complex Hilbert space $\mathcal{H}.$ If $T$ is such a channel and 
$\rho_{{\rm in}} = \rho$ is an input state the channel gives an output 
state$\rho_{{\rm out}} = T(\rho).$ Such a channel is treated as a communication 
resource. If $\mathcal{H} = L^2 (\mathbb{R}^n)$ we say that a channel $T$ is 
{\it Gaussian} if, for every Gaussian state $\rho,$ the output state $T(\rho)$ 
is also Gaussian. The set of all channels in $L^2 (\mathbb{R}^n)$ is a 
semigroup under composition and also a convex set under mixture. The set of all 
Gaussian channels is a subsemigroup of the semigroup of all channels in $L^2 
(\mathbb{R}^n).$ We shall present some examples of Gaussian channels and 
conclude with some open problems.

\begin{example}[Reversible Gaussian Channels]\label{ex3.1}
 If  $U$ is a Gaussian symmetry operator as defined  at the end of property 
(iii) of Gaussian states in Section 2, then $T(\rho)=U\rho U^{-1}$ defines a 
Gaussian channel. $T^{-1} (\rho) = U^{-1} \rho U$ is the reverse of such a 
Gaussian channel with $T^{-1} \circ T = \,{\rm identity}.$ These are the only 
reversible Gaussian channels.
\end{example}

\begin{example}[Bosonic Gaussian channels \cite{2}, \cite{5}]\label{ex3.2}
 Let $\xi_j, \eta_j,$ $1 \leq j \leq n$ be real-valued mean zero random 
variables with a joint normal distribution and let $\zeta_j = \xi_j + i 
\eta_j.$ Write $\mathbf{\zeta}^T = (\zeta_1, \zeta_2, \ldots, \zeta_n)$ and put 
\begin{equation}
T(\rho) = \mathbb{E} W (\mathbf{\zeta}) \,\,\rho \,\,W(\mathbf{\zeta})^{-1} 
\label{eq3.1}
\end{equation}
where $\mathbb{E}$ denotes expectation with respect to the probability 
distribution of $\mathbf{\zeta}$ and $W(\mathbf{\zeta})$ is the Weyl operator 
at $\mathbf{\zeta}.$  If $\rho = \rho_g ({\boldsymbol \ell}, \mathbf{m}, S)$ in 
$L^2 
(\mathbb{R}^n)$ as defined in (\ref{eq2.8}), then (\ref{eq2.4}) implies that 
the quantum Fourier transform of $T(\rho)$ is given by
\begin{eqnarray*}
 T(\rho)^{\wedge} (\mathbf{u}) &=& \mbox{Tr} \,\mathbb{E} \,W (\mathbf{\zeta}) 
\rho_g ({\boldsymbol \ell}, \mathbf{m}, S) W (\mathbf{\zeta})^{-1} 
W(\mathbf{u})\\
&=& \rho_g ({\boldsymbol \ell}, \mathbf{m}, S)^{\wedge} (\mathbf{u}) 
\,\mathbb{E} 
\,e^{2 i \mathcal{I}m \langle \mathbf{\zeta}|\mathbf{u}\rangle}.
\end{eqnarray*}
If $\mathbf{u} = \mathbf{s} + i \mathbf{t}$ where $\mathbf{s}$ and $\mathbf{t}$ 
are in $\mathbb{R}^n$ and the covariance matrix of the Gaussian random vector 
$\sqrt{2} (-\mathbf{\eta}^T, \mathbf{\xi}^T)$ is $C$ then it follows that
$$T (\rho_g ({\boldsymbol \ell}, \mathbf{m}, S)) = \rho_g ({\boldsymbol \ell}, 
\mathbf{m}, S+C).$$
Thus $T$ is a Gaussian channel which changes the covariance matrix $S$ of the 
input Gaussian state to the covariance matrix $S+C$ of the output Gaussian 
state but leaves the mean momentum-position vector unchanged.

If $P$ is the probability distribution of the Gaussian random vector 
$\mathbf{\zeta}$ in (\ref{eq3.1}) and $\{H_{\mathbf{k}}, \mathbf{k} = (k_1, 
k_2, \ldots, k_{2n}) \}$ is the Hermite basis of orthonormal polynomials in 
$L^2 (P),$ define the operators
$$L_{\mathbf{k}} = \mathbb{E}\,\, H_{\mathbf{k}} (\mathbf{\xi}, 
\mathbf{\eta}) W(\mathbf{\zeta}). $$
Then the channel $T$ in (\ref{eq3.1}) can also be written as
$$T(\rho) = \sum_{\mathbf{k}} L_{\mathbf{k}} 
\,\,\rho\,\,L_{\mathbf{k}}^{\dagger}$$
in the Kraus or operator sum form.
\end{example}

\begin{example}[Symplectic Gaussian channels]\label{ex3.3}
 Choose and fix a mean $\mathbf{0}$ Gaussian state $\rho_0$ in$L^2 
(\mathbb{R}^m)$ and fix a symplectic matrix $L \in Sp(2(n+m)).$ For any state 
$\rho$ in $L^2 (\mathbb{R}^n)$ define
\begin{equation}
T(\rho) = \mbox{Tr}_2 \,\,\Gamma (L) (\rho \otimes \rho_0) \Gamma 
(L)^{\dagger}\label{eq3.2}
\end{equation}
where $\Gamma(L)$ is the Gaussian symmetry operator associated with $L$ and 
$\mbox{Tr}_2$ is the relative trace over the component $L^2 
(\mathbb{R}^m)$ in $L^2 (\mathbb{R}^n) \otimes L^2 (\mathbb{R}^m) = L^2 
(\mathbb{R}^{n+m}).$ Since conjugation by a unitary operator and relative trace 
are completely positive and trace preseving, it follows that $T$ as defined in 
(\ref{eq3.2}) is a channel. If  $\rho$ is Gaussian it follows from Example 
\ref{ex3.1} and property (i) of Gaussian states in Section 2 that $T$ is 
Gaussian. We call $T$ in (\ref{eq3.2}) a {\it symplectic Gaussian channel}.

We shall now examine how a {\it symplectic Gaussian channel} changes the mean 
and covariance parameters of a Gaussian state. By the Gaussian purificaton 
property (vi) of a Gaussian state in Section 2, $L^2 (\mathbb{R}^m)$ can be 
replaced by $L^2 (\mathbb{R}^m) \otimes L^2 (\mathbb{R}^m)$ and $\rho_0$ in 
(\ref{eq3.2}) by a pure Gaussian state which is determined by a unit vector of 
the form $\Gamma (M) \,\,|e (\mathbf{0})\rangle$ in $L^2 (\mathbb{R}^{2m})$ 
with $M \in Sp (2m).$ Put $k = 2m$ and note that $T$ in (\ref{eq3.2}) can as 
well be replaced by
\begin{equation}
 T(\rho) = \mbox{Tr}_2 \,\,\Gamma (L) \,\,\rho \otimes | e (\mathbf{0})\rangle 
\langle e (\mathbf{0})| \,\,\Gamma (L)^{\dagger} \label{eq3.3}
\end{equation}
with a different $L$ in $Sp(2(n+k)).$ Note that
\begin{eqnarray*}
| e (\mathbf{0})\rangle \langle e (\mathbf{0})| &=& \rho_g (\mathbf{0}, 
\mathbf{0}, \frac{1}{2} I_{2k}),   \\
\rho_g ({\boldsymbol \ell}, \mathbf{m}, S) \otimes \rho_g (\mathbf{0}, 
\mathbf{0}, 
\frac{1}{2} I_{2k}) &=& \rho_g ({\boldsymbol \ell} \oplus \mathbf{0}, 
\mathbf{m} 
\oplus \mathbf{0}, S \oplus \frac{1}{2} I_{2k}).
\end{eqnarray*}
By the transformation property of Gaussian states we have from (\ref{eq2.3})
$$T (\rho_g ({\boldsymbol \ell}, \mathbf{m}, S)) = \rho_g  
({\boldsymbol \ell}^{\prime} 
\oplus \mathbf{a}, \mathbf{m}^{\prime} \oplus \mathbf{b}, S^{\prime})$$
where 
$$S^{\prime} = (L^{-1})^T \,\,\left [\begin{array}{cc} S & 0 \\ 0 & \frac{1}{2} 
I_{2k}   \end{array} \right ] \,\,L^{-1} $$
and ${\boldsymbol \ell}^{\prime},$ $\mathbf{m}^{\prime},$ $\mathbf{a},$ 
$\mathbf{b}$ 
are obtained as follows. Define $\mathbf{\mu}, \mathbf{\mu}^{\prime}$ by 
$\mathbf{\mu}^T = ({\boldsymbol \ell}_1, - m_1, \ldots, \ell_n, - 
m_n, 0, 0, \ldots, 0),$ $(\mathbf{\mu}^{\prime})^T = (\ell_1^{\prime}, 
-m_1^{\prime}, \ldots, \ell_n^{\prime},  -m_n^{\prime}, a_1, -b_1, \ldots, 
a_k, -b_k).$ Then
$$\mathbf{\mu}^{\prime} = (L^{-1})^T\,\, \mathbf{\mu} \quad {\rm in} 
\quad \mathbb{R}^{2(n+k)}. $$
Writing
$$L^{-1} = M = \left[\begin{array}{cc}M_{11} & M_{12} \\ M_{21} & M_{22}  
\end{array}  \right ], \quad S^{\prime} = 
\left[\begin{array}{cc}S^{\prime}_{11} & 
S^{\prime}_{12} \\ S^{\prime}_{21} & S^{\prime}_{22}  
\end{array}  \right ] $$
in block notation where $11$ and $22$ blocks are of order $2n \times 2n$ and 
$2k \times 2k$ respectively we see that, for $T$ as in (\ref{eq3.3}),
$$ T (\rho_g ({\boldsymbol \ell}, \mathbf{m}, S)) = \rho_g 
(({\boldsymbol \ell}^{\prime}, \mathbf{m}^{\prime}, S_{11}^{\prime})$$
where
\begin{eqnarray*}
 S_{11}^{\prime} &=& M_{11}^T S M_{11} + \frac{1}{2} M_{21}^T M_{21},\\
\left [\begin{array}{c} \ell_1^{\prime} \\ -m_{1}^{\prime} \\ \vdots \\ 
\ell_n^{\prime} \\ -m_{n}^{\prime} \end{array} \right ] &=&  M_{11}^T \left 
[\begin{array}{c} \ell_1 \\ -m_{1} \\ \vdots \\ 
\ell_n \\ -m_{n} \end{array} \right ]. 
\end{eqnarray*}
We summarize our algebra in the form of a theorem.
\end{example}

\begin{theorem}\label{thm3.4}
The most general symplectic Gaussian channel $T$ has the property that for all 
Gaussian states $\rho_g ({\boldsymbol \ell}, \mathbf{m}, S)$ in 
$L^2(\mathbb{R}^n),$
$$ T (\rho_g ({\boldsymbol \ell}, \mathbf{m}, S)) = \rho_g 
(({\boldsymbol \ell}^{\prime}, 
\mathbf{m}^{\prime}, S^{\prime})$$
where
\begin{equation}
S^{\prime} = M_{11}^T \,\, S \,\,M_{11} + \frac{1}{2} \,\,M_{21}^T \,\, M_{21} 
\label{eq3.4}
\end{equation}
for some $M \in Sp (2 (n+k))$ for some $k$ with
\begin{equation}
 M = \left [\begin{array}{cc} M_{11} & M_{12} \\ M_{21} & M_{22} \end{array} 
\right ] \label{eq3.5}
\end{equation}
where the blocks with  labels $11$ and $22$ are matrices of order $2n \times 
2n$ and $2k \times 2k$ respectively and the vectors $\mathbf{\mu}, 
\mathbf{\mu}^{\prime}$ defined by
\begin{eqnarray}
 \mathbf{\mu}^T &=& (\ell_1, -m_1, \ldots, \ell_n, -m_n), \label{eq3.6}\\
 \mathbf{\mu}^{\prime^{T}} &=& (\ell_1^{\prime}, -m_1^{\prime}, \ldots, 
\ell_n^{\prime}, -m_n^{\prime}), \label{eq3.7}
\end{eqnarray}
where
\begin{equation}
\mathbf{\mu}^{\prime} = M_{11}^T \,\,\mathbf{\mu}.   \label{eq3.8}
\end{equation}
\end{theorem}
We shall denote by $T_M$ any symplectic Gaussian channel in $L^2 
(\mathbb{R}^n)$ 
satisfying Theorem \ref{thm3.4} for some $k,$ some $M \in Sp (2 (n+k))$ and 
equations (\ref{eq3.4})-(\ref{eq3.8}) for all Gaussian states $\rho_g 
({\boldsymbol \ell}, \mathbf{m}, S).$

\begin{theorem}\label{thm3.5}
Let $T_L,$ $T_M$  be symplectic Gaussian channels in $L^2(\mathbb{R}^n)$ where 
$L \in Sp (2(n+\ell)),$ $M \in Sp (2(n+m)).$ Then there exists a symplectic 
Gaussian channel $T_N$ in $L^2 (\mathbb{R}^n)$ for some $N \in Sp 
(2(n+\ell+m))$ such that for every Gaussian state $\rho$ in $L^2 
(\mathbb{R}^n)$
$$T_N (\rho) = T_L (T_M (\rho)).$$
\end{theorem}

\begin{proof}
 Express the matrices $L$ and $M$ in block notation
$$L = \left [\begin{array}{cc} L_{00} & L_{01} \\ L_{10} & L_{11}  \end{array} 
\right ], \quad M = \left [\begin{array}{cc} M_{00} & M_{02} \\ M_{20} & M_{22} 
 \end{array} \right ] $$
where $L_{00}$ and $M_{00}$ are of order $2n \times 2n,$ $L_{11}$ is of order 
$2\ell \times 2 \ell$ and $M_{22}$ is of order $2m \times 2m.$ Define
\begin{eqnarray*}
\widetilde{L} = \left [\begin{array}{ccc} L_{00} & L_{01} & 0 \\ L_{10} & 
L_{11} & 0 \\ 0 & 0 & I_{2m}  \end{array}  \right ], && \widetilde{M} = \left 
[\begin{array}{ccc} M_{00} & 0 & M_{02} \\ 0 & I_{2 \ell} & 0 \\ M_{20} & 0 & 
M_{22}  \end{array} \right ],\\
N = \widetilde{M} \widetilde{L} &=& \left [\begin{array}{ccc} M_{00} L_{00} 
& M_{00} L_{01} & M_{02} \\ L_{10} & L_{11} & 0 \\ M_{20} L_{00} & M_{20} 
L_{01} & M_{22}   \end{array} \right ].
\end{eqnarray*}
Then $\widetilde{M}, \widetilde{L}$ and $N$ are all elements of $Sp 
(2(n+\ell+m)).$ Consider the symplectic Gaussian channel $T_N$ as described in 
Theorem \ref{thm3.4}. Then
\begin{equation}
T_N (\rho_g ({\boldsymbol \ell}, \mathbf{m}, S)) = \rho_g 
({\boldsymbol \ell}^{\prime}, 
\mathbf{m}^{\prime}, S^{\prime})    \label{eq3.9}
\end{equation}
where
\begin{eqnarray*}
S^{\prime} &=& (M_{00} L_{00})^T S \,\,M_{00}L_{00} + \frac{1}{2} (L_{10}^T 
L_{10} + (M_{20} L_{00})^T M_{20} L_{00})\\
&=& L_{00}^T \left \{M_{00}^T  S M_{00} + \frac{1}{2} M_{20}^T M_{20} \right 
\} L_{00} + \frac{1}{2} L_{10}^T L_{10}
\end{eqnarray*}
which is also the covariance matrix of the Gaussian state $T_L (T_M(\rho_g 
({\boldsymbol \ell}, \mathbf{m}, S))).$ Since $N$ is a symplectic dilation of 
$M_{00} L_{00},$ equations (\ref{eq3.6})-(\ref{eq3.8}) in Theorem \ref{thm3.4} 
imply that the momentum and position means ${\boldsymbol \ell}^{\prime}, 
\mathbf{m}^{\prime}$ in (\ref{eq3.9}) agree with the momentum and position 
means of $T_L (T_M(\rho_g({\boldsymbol \ell}, \mathbf{m}, S))).$ \hfill{\qed}
\end{proof}

\begin{corollary}\label{cor3.6}
Let $A,B$ be real $2n \times 2n$ matrices admitting symplectic dilations of 
order $2(n+\ell),$ $2(n+m)$ respectively. Then $AB$ admits a symplectic 
dilation of order $2(n + \ell + m).$
\end{corollary}

\begin{proof}
 This is seen immediately from the proof of Theorem \ref{thm3.5} if we write $A=M_{00},$ 
$B=L_{00}.$ \hfill{\qed}
\end{proof}

\begin{remark}\label{rem3.7}
From equation (\ref{eq3.3}) we can easily write down a Kraus or operator sum 
representation of the symplectic Gaussian channel. Indeed, in  (\ref{eq3.3}) 
put $U = \Gamma (L)$ where $L \in Sp (2(n+k)).$ Consider the particle number 
basis $\{| r_1, r_2, \ldots, r_k \rangle, r_j \in \{ 0,1,2, \ldots, \} \forall 
\,j  \}$ in $L^2 (\mathbb{R}^k)$ when identified with the boson Fock space over 
$\mathbb{C}^k.$ Define the operators $U_{\mathbf{r}, \mathbf{s}}$ in $L^2 
(\mathbb{R}^n)$  by the identity
$$\langle \psi \otimes \mathbf{r}| U | \psi^{\prime} \otimes \mathbf{s} \rangle 
= 
\langle \psi | U_{\mathbf{r}, \mathbf{s}} | \psi^{\prime} \rangle \quad 
\forall \,\,\psi, \psi^{\prime} \in L^2 (\mathbb{R}^n).$$
Then
\begin{eqnarray*}
 T(\rho) &=& \mbox{Tr}_2 \,\,\Gamma (L) \,\,\left ( \rho \otimes 
|\mathbf{0}\rangle \langle \mathbf{0} |) \Gamma (L \right )^{\dagger}       
\\  
 &=&  \mbox{Tr}_2 \,\, U (\rho \otimes | \mathbf{0}\rangle \langle \mathbf{0} 
|) U^{\dagger}   )    \\   
 &=&  \sum_{r} \,\,U_{\mathbf{r}, \mathbf{0}} \,\,\rho\,\, U_{\mathbf{r}, 
\mathbf{0}}^{\dagger}    
\end{eqnarray*}
\end{remark}

\begin{example}[Quasifree channels\cite{3}, \cite{9}]\label{ex3.8}
This example is from the construction given by Heinosaari, Holevo and Wolf 
\cite{3}. To describe this we rewrite the Weyl operators $W(\mathbf{u})$ as 
$W(\mathbf{\xi})$ where  $\mathbf{\xi}^T = (\xi_1, \xi_2, \ldots, \xi_{2n}) = 
(\mbox{Re}\, u_1, \mbox{Im}\,u_1, \ldots,\mbox{Re}\, \,u_n,  \mbox{Im}\,\,u_n)$ 
for $\mathbf{u} \in \mathbb{C}^n.$ Then there exists a unital completely 
positive map $T_0$ on $\mathcal{B}(L^2(\mathbb{R}^n))$ satisfying
\begin{equation}
T_0 (W(\mathbf{\xi})) = e^{-\frac{1}{2} \mathbf{\xi}^{T} B 
\mathbf{\xi}} \,\,W(A \mathbf{\xi}) , \quad   \mathbf{\xi} \in 
\mathbb{R}^{2n}  \label{eq3.10}
\end{equation}
whenever $A$ and $B$ are $2n \times 2n$ real matrices, $B$ is symmetric and the 
matrix inequality
\begin{equation}
B + i (A^T J_{2n} A - J_{2n}) \ge 0   \label{eq3.11}
\end{equation}
holds with $J_{2n}$ given by (\ref{eq1.1}). The left hand side of 
(\ref{eq3.11}) is a complex hermitian matrix and (\ref{eq3.11}) implies that $B 
\ge 0.$

Now choose and fix $A,B$ as above and consider $T_0$ satisfying (\ref{eq3.10}). 
For any state $\rho$ in $L^2(\mathbb{R}^n)$ define the state $T(\rho)$ by
\begin{eqnarray}
\mbox{Tr}\,\,T(\rho) \,\,W(\mathbf{\xi}) &=&  \mbox{Tr}\,\, \rho \,\,T_0 (  
W(\mathbf{\xi})) \nonumber \\
&=& \mbox{Tr}\,\, \rho \,\,  W(A \mathbf{\xi}) e^{-\frac{1}{2} \mathbf{\xi}^{T} 
B \mathbf{\xi}} \quad \forall \,\,\mathbf{\xi}.  \label{eq3.12}
\end{eqnarray}
Then for any Gaussian state $\rho_g ({\boldsymbol \ell}, \mathbf{m}, S)$ we 
have 
from (\ref{eq2.7})
$$\mbox{Tr} \,\,T (\rho_g ({\boldsymbol \ell}, \mathbf{m}, S)) = \rho_g 
({\boldsymbol \ell}^{\prime}, \mathbf{m}^{\prime}, S^{\prime}) $$
where 
$$S^{\prime} = A^T SA + \frac{1}{2} B$$
and ${\boldsymbol \ell}^{\prime}, \mathbf{m}^{\prime}$ are given by
$$ A^T \left [\begin{array}{c} \ell_1 \\ -m_1 \\ \vdots \\ \ell_n \\ -m_n 
\end{array} \right ] = \left [\begin{array}{c} \ell_1^{\prime} \\ -m_1^{\prime} 
\\ \vdots \\ \ell_n^{\prime} \\ -m_n^{\prime} \end{array} \right ]$$
Thus $T$ is a Gaussian channel which changes the means and the covariance 
matrix exactly in the same manner as for the symplectic Gaussian channel $T_M$ 
of Theorem \ref{thm3.4} by writing $M_{11} = A$ and $M_{21}^{\dagger} M_{21} = 
B,$ associated with the symplectic matrix $M.$ We call the channel defined 
through (\ref{eq3.10}) and (\ref{eq3.11}), a {\it quasifree Gaussian channel}.

The inequality (\ref{eq3.12}) raises some questions concerning symplectic 
dilations. To any $2n \times 2n$ real matrix $A,$ associate the convex sets
\begin{eqnarray*}
 \mathcal{K}_n &=& \left \{ S | S \ge 0, \quad 2 S -i J_{2n} \ge 0 \right \}, \\
\mathcal{F}_n(A) &=& \left \{B | B \ge 0, A^T S A + \frac{1}{2} B \ge 0 \quad 
\forall \,\,S \in \mathcal{K}_n \right \}, \\
\mathcal{F}_n^0(A) &=& \left \{B |B \ge 0, i (A^T J_{2n} A - J_{2n}) + B \ge 0 
\right \}.
\end{eqnarray*}
By Theorem \ref{thm2.3}, $\mathcal{K}_n$ is the set of all $2n \times 2n$ 
covariance 
matrices of Gaussian states, $B \in \mathcal{F}_n (A)$ if and only if the 
affine tranformaton $S \rightarrow A^T S A + \frac{1}{2} B$ leaves 
$\mathcal{K}_n$ 
invariant and $\mathcal{F}_n^0(A)$ is the set of all $2n \times 2n$ positive 
definite matrices such that  $(A,B)$ defines a quasifree Gaussian channel. since
\begin{eqnarray*}
\lefteqn{2 \left (A^T \,\,\,S\,\,\,A + \frac{1}{2}B \right ) - i \,\,\,J_{2n}}\\
&=& A^T \left (2S - i\,\,\, J_{2n} \right )\,\,\, A + i \,\,\,\left (A^T 
\,\,\,J_{2n} A - J_{2n}\right ) + \,\,\, B
\end{eqnarray*}
it follows that $\mathcal{F}_n^0 (A) \subset \mathcal{F}_n (A).$ Is it true 
for every $B$ in $\mathcal{F}_n (A)$ there is a Gaussian channel with the 
property that it transforms a Gaussian state $\rho_g( {\boldsymbol \ell}, 
\mathbf{m}, S)$ to a Gaussian state with covariance matrix $A^T SA + 
\frac{1}{2} B?$ To any $B \in \mathcal{F}_n^0 (A)$ does there exist a 
symplectic dilation $\widetilde{A} = \left [\begin{array}{cc} A & P \\ Q & R 
\end{array} \right ]$ such that $B = Q^T Q?$ If this holds we can realize the 
quasifree channel associated with $(A,B)$ by a symplectic channel associated 
with $\widetilde{A}.$ If $\left [\begin{array}{cc} A & P \\ Q & R 
\end{array} \right ]$ is a symplectic matrix does $Q^T Q \in 
\mathcal{F}_n^0(A)?$ Finally, are there Gaussian channels not belonging to the 
semigroup generated by all reversible, bosonic, symplectic and quasifree 
Gaussian channels? It would be interesting to find answers to all the questions 
raised above. One would also like to have a description of the extreme points of 
$\mathcal{F}_n(A)$ and $\mathcal{F}_n^0(A).$ 
\end{example}

\vskip0.1in
\section*{Acknowledgement} This paper is an expanded version of a lecture 
delivered in the Kerala School of Mathematics, Kozhikode at a workshop and 
conference held in honour of Professor Kalyan B. Sinha on his 70th birthday 
during 7-14 February 2014. I thank the faculty and staff of KSOM for their warm 
hospitality in a scenic campus surrounded by hills and native trees. Part of 
this work was done at the Institute of Mathematical Sciences, Chennai where I 
enjoyed their warm hospitality and benefitted from several conversations with 
M. Krishna and R. Simon. I am grateful to Ajit Iqbal Singh for a careful 
reading of the manuscript and suggesting many improvements.   I thank my 
colleagues at the Delhi Centre of the Indian Statistical Institute for providing 
me a friendly research environment.


\end{document}